\documentclass[aps,pra,10pt,notitlepage,twocolumn]{revtex4-2}
\usepackage{amsfonts,amssymb,amsmath}            
\usepackage{lmodern}
\usepackage[]{graphics,graphicx}            
\usepackage{amsthm}
\usepackage{bbold,enumitem,mathtools}
\usepackage{booktabs}
\usepackage{physics}
\usepackage{tikz}
\usepackage{tikz-3dplot}
\usetikzlibrary{decorations.pathmorphing,patterns,decorations.markings,matrix,quantikz2}
\usetikzlibrary{arrows}
\usetikzlibrary{bending}
\usetikzlibrary {backgrounds} 
\usepackage{adjustbox}
\usepackage[breakable]{tcolorbox}
\usepackage{xstring}
\usepackage{ifthen}
\usepackage{listings}
    \usepackage{hyperref}
\usepackage[capitalize]{cleveref}

\usepackage[english]{babel}

\makeatletter
\def\bbl@set@language#1{%
  \edef\languagename{%
    \ifnum\escapechar=\expandafter`\string#1\@empty
    \else\string#1\@empty\fi}%
  \@ifundefined{babel@language@alias@\languagename}{}{%
    \edef\languagename{\@nameuse{babel@language@alias@\languagename}}%
  }%
  \select@language{\languagename}%
  \expandafter\ifx\csname date\languagename\endcsname\relax\else
    \if@filesw
      \protected@write\@auxout{}{\string\select@language{\languagename}}%
      \bbl@for\bbl@tempa\BabelContentsFiles{%
        \addtocontents{\bbl@tempa}{\xstring\select@language{\languagename}}}%
      \bbl@usehooks{write}{}%
    \fi
  \fi}
\newcommand{\DeclareLanguageAlias}[2]{%
  \global\@namedef{babel@language@alias@#1}{#2}%
}
\makeatother

\DeclareLanguageAlias{en}{english}

\newtheorem{lemma}{Lemma}
\newtheorem{corollary}{Corollary}
\newtheorem{assumption}{Assumption}
\newenvironment{proof*}[1][\proofname]{%
  
  \begin{proof}[#1]}{\end{proof}}

\newcommand{\identity}{\mathbb{1}}
\renewcommand{\epsilon}{\varepsilon}
\newtcolorbox{example}[1][]{width=0.48\textwidth,boxrule=0mm,leftrule=1mm,colframe=black!75,sharp corners,before=\par\smallskip\centering,after=\par,title=#1,breakable=true}

\forceredefine{}

\begin{document}

\title{Transversal Gates for Highly Asymmetric qLDPC Codes}
\date{\today}
\author{Heather \surname{Leitch}}
\author{Alastair \surname{Kay}}
\email{alastair.kay@rhul.ac.uk}
\affiliation{Royal Holloway University of London, Egham, Surrey, TW20 0EX, UK}
\begin{abstract}
Transversal gates are the ideal gates in a fault-tolerant scenario; relatively easy to implement, and minimally error propagating. Their availability will maximise fault tolerant thresholds, enabling universal quantum computation in a wider range of noisy hardware.

Transversal gates in quantum low density parity check (qLDPC) codes are largely unstudied, with the early results of Burton \& Browne suggesting that transversal non-Clifford gates may be impossible. In this paper, we contradict this expectation with constructions for both hypergraph product codes and balanced product codes, although these first examples have weak properties. We find qLDPC codes with transversal phase gates that have a number of logical qubits that grows linearly with $n$, the number of physical qubits. The distance is highly asymmetric; while the distance against bit flip errors also grows (almost) linearly in $n$, the distance against phase flip errors is limited to $\order{1}$. Moreover, existing distance rebalancing techniques are one-sided; we show that they only preserve the transversality when rebalancing to increase the asymmetry, not decrease it.

We also collate a toolbox of techniques that identify a single transversal gate from which we can individuate a large range of different transversal gates. This is critical when addressing the question of what a transversal phase gate truly means when there are many logical qubits in the system.


\end{abstract}
\maketitle

The greatest challenge towards realising a fully functional quantum computer is the management of noise. Existing devices, while containing hundreds of qubits, are very noisy, severely curtailing their useful operating time, and consequently the utility of their computations. Experimentally, we are on the verge of the tipping point wherein error correction provides a gain in quality \cite{sivak2023}. The ultimate target is clear -- a fault-tolerant quantum computer. We know that there are methods, such as concatenation \cite{aliferis2006,aliferis2007,paetznick2013a,knill2005} and surface codes \cite{dennis2002,breuckmann2016a}, that can guarantee the existence of a fault-tolerant threshold; a critical error rate below which computation can proceed for arbitrarily long times with high accuracy.

It remains unclear whether either of these options is the best long-term solution, and it is vital to study alternative methods that may reduce the overheads, particularly in terms of the number of qubits involved, so that we can achieve more with the limited number of qubits that we have available. When we design error correcting codes for a quantum computer, there are a number of desirable properties which are motivated by the demand for a high fault-tolerant threshold and ease of implementation:
\begin{itemize}
\item High density, robust storage: the number of stored qubits $k$ and the code distance $d$ should be as large as possible relative to $n$, the number of physical qubits, in an $[[n,k,d]]$ code.
\item Local: parity checks should only act on qubits which are physically close to each other.
\item Low density: parity checks should only involve a small number of qubits, and each qubit should only be involved in a small number of parity checks.
\item Universal Gate Set: unitaries should be implemented in such a way as to minimise the chance for errors to enter the system. Typically that might involve having as large a transversal gate set as possible, subject to the constraints of the Eastin-Knill theorem \cite{eastin2009}, conveying that a universal set of transverse gates is impossible.
\item Over-complete Gate Set: universality for error correcting codes necessarily comprises a finite set of gates that can synthesise other unitaries to arbitrary accuracy \cite{dawson2005,kliuchnikov2012,kliuchnikov2015}. The more gates that we have available beyond the bare minimum to achieve universality, the easier synthesis will be.
\end{itemize}
These requirements are not all mutually compatible. In particular, imposing that a code be local limits the storage density and available transversal gate sets -- a code that is local in $D$ dimensions cannot exceed $kd^\frac{2}{D-1}\sim n$, as saturated by the $D=2$ Toric code \cite{kitaev1997}, and the transversal gates can only be drawn from the $(D-1)\textsuperscript{th}$ level of the Clifford hierarchy \cite{bravyi2013,bravyi2010b}. However, new architectures are emerging in which the requirement for local gates is not as pressing \cite{bluvstein2024}.

Within the class of quantum Low Density Parity Check (qLDPC) codes, formed by removing the locality restriction while continuing to impose low density, vastly improved rates are known to exist, culminating in good codes with $k,d\sim n$ \cite{tillich2014,breuckmann2021,panteleev2022,panteleev2022a} and linear time decoding \cite{dinur2022}. Thus, at least in the asymptotic limit of size, the density must vastly out-perform both concatenated codes \footnote{without further improvements, such as \cite{yamasaki2024}.} and surface codes. Even at more modest sizes, the properties are promising \cite{bravyi2024}.

To date, little is known about the ability to perform fault-tolerant computation on qLDPC codes, other than some very general results which lack the improvements that can be yielded by specialisation \cite{gottesman2014,kovalev2013,fawzi2018}. The availability of transversal gates in these systems would have a massive impact upon fault tolerant thresholds and yet, to date, is a complete unknown. Special cases of code constructions have been shown not to have transversal gates beyond the Clifford gates \cite{burton2022}, and this has encouraged subsequent studies to take alternative routes to creating logical gates \cite{quintavalle2023,krishna2021}, such as code rewiring \cite{banfield2022}.

In this paper, we reverse this trend and show that it is indeed possible to construct qLDPC codes with transversal gate sets beyond the Clifford gates. This construction should be broadly applicable, and we demonstrate it for both balanced product code \cite{breuckmann2021,leitch2025} and the hypergraph product code \cite{tillich2014}. In \cref{sec:transversal}, we review a variation of the formalism of \cite{webster2023} that allows us to recognise whether a given code has the transversal phase gates that we're after. We establish what we mean by transversal gates in \cref{sec:toolbox}, especially in the context of a code with multiple logical qubits. We demonstrate that it is not necessary to actively produce all the individual phase gates on each logical qubit. Instead, it is good enough to find a single gate that acts transversally at the logical level, from which we will be able to extract the individual behaviour we need. This sets us on the road toward our main construction in \cref{sec:balanced_transverse} in which we use the properties of a local code in the Tanner code construction to engineer the transversal gate properties that we desire.

\section{Transversal Phase Gates}\label{sec:transversal}

We are interested in phase gates
$$
P_q=\begin{bmatrix} 1 & 0 \\ 0 & e^{2\pi i/2^q}\end{bmatrix}
$$
for positive integer $q$. Familiar instances are $Z,S,T$ with $q=1,2,3$ respectively. We assert that $q=4$ should be the primary target (see \cref{sec:toolbox}), while smaller values of $q$ will provide a path towards implementation. To differentiate the phases applied to $n$ different qubits, we take a string $p\in[2^q]^n$, with $[2^q]$ indicating the set of integers 0 to $2^q-1$ inclusive, such that
$$
P_{q}^{p}=\bigotimes_{i=1}^nP_q^{p_i}.
$$
Our aim is to apply many of these gates to the $n$ different physical qubits of a system, and realise some of these gates at the logical level. In other words,
\begin{equation}\label{eq:transverse_vs_physical}
P_{L,q}^w\equiv P_q^p
\end{equation}
where $w\in[2^q]^k$ acts on the $k$ logical qubits.

The necessary and sufficient conditions for a CSS code to have a transversal $P_q$ gate were essentially found in \cite{koutsioumpas2022,webster2023}. We restate their results more succinctly, making use of the notation that
$$
x\cdot y=(x_1y_1) (x_2y_2) \ldots (x_ny_n).
$$
If $x,y\in\{0,1\}^n$, this returns the bit string that indicates the positions where $x,y$ are both 1.
\begin{example}
The dot product works as
$$
0011\cdot 1010=0010.
$$
while the weight yields the usual inner product,
$$
|0011\cdot 1010|=1.
$$
\end{example}
The weight of the bit string $x$ is $|x|$. The dot product is associative, thus extending to triple products
$$
x\cdot y\cdot z=(x\cdot y)\cdot z
$$
and beyond. We also apply it to matrices, e.g.\ $H^{\cdot 3}$
is all possible triple-products of distinct rows of $H$, but we can choose to retain only the linearly independent ones (and by keeping only distinct rows, we are already removing terms from lower order, $H$ and $H^{\cdot 2}$).

\begin{example}
If
$$
H=\begin{bmatrix} 0 & 0 & 1 & 1 \\ 1 & 0 & 1 & 0 \end{bmatrix},
$$
then all possible inner products are
\begin{align*}
0011\cdot0011 &=0011 \\
1010\cdot 1010&=1010 \\ 
0011\cdot 1010&=0010.
\end{align*}
Since the first two are already contained in $H$,
$$
H^{\cdot 2}=\begin{bmatrix} 0 & 0 & 1 & 0 \end{bmatrix}.
$$
The weight operator applies to rows of a matrix:
$$
|H|=\begin{bmatrix}2\\2\end{bmatrix}.
$$
\end{example}

The dot product is tightly connected with addition modulo 2:
\begin{align}
x\oplus y&=x+y-2x\cdot y \label{eq:xor}\\
x\oplus y\oplus z&=x+y+z-2(x\cdot y+x\cdot z+y\cdot z)+4x\cdot y\cdot z \nonumber\\
\bigoplus_ix_i&=\sum_ix_i-2\sum_{i\neq j}x_i\cdot x_j+4\sum_{i\neq j\neq k}x_i\cdot x_j\cdot x_k-\ldots \nonumber
\end{align}
As one approaches larger additions, the leading coefficient increases by a factor of 2, and the sign alternates. 
\begin{example}$$
0011\oplus 1010=\adjustbox{valign=b}{\begin{tabular}{cr}
&0011\\ +&1010 \\-&0020\\\bottomrule &1001
\end{tabular}}
$$
\end{example}

Consider a binary matrix $H_X$ of $m$ rows. We can write
\begin{align}
H_X^{\oplus}&=\bigoplus_{i=1}^mr_i \nonumber\\ 
&=\sum_{i=1}^m(-1)^{i+1}2^{i-1}\left(\sum_{r\in H_X^{\cdot m}}r\right) \label{eq:opls_matrix_id}
\end{align}
where the $r_i$ are distinct rows of $H_X$. Note that if we want to evaluate $|H_X^{\oplus}|\text{ mod }2^q$, then we only need to consider terms of the sum $i\leq q$.

An error correcting code is specified by binary matrices that describe the parity check operations that identify any errors. In particular, the CSS codes, with which we work exclusively, give two binary matrices $H_X$ and $H_Z$, both with $n$ columns, and for which $H_XH_Z^T=0\text{ mod }2$. For example, if $r$ is a row of $H_X$, then there is a parity check that we will measure comprising
$$
X^r=\bigotimes_{i=1}^nX^{r_i}.
$$
It applies $X$ on the qubits $i$ where $r_i=1$, and does nothing to the others. If the code encodes $k$ logical qubits, then there are also $k$ logical $X$-type operators $L_X$ and $k$ $Z$-type operators $L_Z$ where
$$
L_XL_Z^T=\identity_k\text{ mod }2.
$$

\begin{lemma}\label{lem:transversal}
For a CSS code on $N$ qubits with $X$ parity check matrices $H_X$, and logical operators $L_X$, the code has transversal $P_{L,q}^w$ if there exists a $p\in[2^q]^n$ such that
\begin{equation}
|H_X^{\cdot i}\cdot L_X^{\cdot j}\cdot p|=0\text{ mod }2^{q+1-i-j}\qquad \forall i+j=1,2,\ldots, q \label{eq:stabilizer_weights}
\end{equation}
except for the case $(i,j)=(0,1)$, for which
\begin{equation}
|L_X\cdot p|=w\text{ mod }2^{q}. \label{eq:logical_weights}
\end{equation}

\end{lemma}
\begin{proof}
Consider a logical state $v\in\{0,1\}^k$ acted on by our transversal gate, \cref{eq:transverse_vs_physical}:
$$
P_{L,q}^w\ket{v}_L=e^{2\pi i|v\cdot w|/2^q}\ket{v}_L.
$$
We aim to realise this logical operation by implementing some $P_q^{p}$ at the physical level, meaning this relation must hold for each of the basis vectors of $\ket{v}_L$ independently:
$$
P^{p}_q\ket{v L_X+u H_X}=e^{2\pi i|v\cdot w|/2^q}\ket{v L_X+u H_X}
$$
for all $v\in\{0,1\}^k$ and $u\in\{0,1\}^m$ if $H_X$ has $m$ rows. Thus, for all $u,v$, we have
$$
|(v L_X\oplus u H_X)\cdot p|\equiv |v\cdot w|\text{ mod }2^q.
$$
In fact, it was shown in \cite{webster2023} that we don't need to test this for all possible bit strings. It is sufficient to only consider strings $uv$ of total weight $\leq q$, making the calculation much more efficient. The goal would be to find a $p\in[2^q]^n$ that satisfies all these equations.

However, we can push this a little further. For a given $u,v$, let us combine all the rows of $H_X,L_X$ corresponding to $u,v=1$ as a matrix $H_{uv}$. Then this condition is just
$$
|H_{uv}^{\oplus}\cdot p|=|v\cdot w|\text{ mod }2^q.
$$
Expanding this using \cref{eq:opls_matrix_id}, we see that for all $u,v$ the $|v\cdot w|$ term is taken care of using
$$
|L_X\cdot p|=w\text{ mod }2^{q}.
$$
Then is it sufficient that all other terms are 0.
\end{proof}

\begin{example}
Consider the Steane code \cite{steane1996a}, which has
\begin{align*}
H_X&=\begin{bmatrix} 1 & 0 & 0 & 1 & 0 & 1 & 1 \\
0 & 1 & 0 & 1 & 1 & 0 & 1 \\
0 & 0 & 1 & 0 & 1 & 1 & 1
\end{bmatrix}\\
L_X&=\begin{bmatrix} 1 & 1 & 1 & 1 & 1 & 1 & 1 \end{bmatrix}.
\end{align*}
Does this code have transversal $S$? The choice of gate has fixed with $q=2$. Assume that $p=\begin{bmatrix} 1 & 1 & 1 & 1 & 1 & 1 & 1 \end{bmatrix}$. Now apply \cref{lem:transversal} for $i+j\leq 2$.
\begin{itemize}
\item For $i+j=1$, one case is $i=1,j=0$. One can readily verify that $H_Xp\equiv 0\text{ mod }4$.
\item The only other option is $i=0,j=1$, the special case. Since $L_X=p$, $|L_X\cdot p|=7$ (this means that by applying $S$ on every physical qubit, we'd be expecting $S^\dagger$ at the logical level).
\item For $i+j=2$, since $L_X$ comprises a single row, $L_X^{\cdot 2}$ is trivial ($i=0,j=2$).
\item In the case $i=j=1$, $|L_X\cdot H_X\cdot p|=|H_X\cdot p|\equiv 0\text{ mod }2$ reduces to a previous case.
\item This only leaves the case $i=2$, $j=0$. The linearly independent terms in $H_X^{\cdot 2}$ are
$$
H_X^{\cdot 2}\equiv \begin{bmatrix}
0 & 0 & 0 & 1 & 0 & 0 & 1 \\
0 & 0 & 0 & 0 & 1 & 0 & 1 \\
0 & 0 & 0 & 0 & 0 & 1 & 1
\end{bmatrix}
$$
such that $H_X^{\cdot 2}p\equiv 0\text{ mod }2$.
\end{itemize}
We conclude that applying $S$ on every physical qubit of the code creates logical $S^\dagger$.

What about transversal $T$? Terms in $H_X^{\cdot 3}$ localise individual qubits. Hence $|H_X^{\cdot 3}\cdot p|\equiv 0\text{ mod }2$ can only be satisfied if $p$ is even on every site, and that cannot yield a logical $T$ since $L_X\cdot p$ needs to be odd. This code does not have transversal $T$.
\end{example}

We will reduce the conditions of \cref{lem:transversal} by making the following assumption:
\begin{assumption}\label{assum:unsupported_logicals}
There exists a representation of the generators of the logical $X$ operators $L_X$ which all act on distinct sets of physical qubits, i.e.
$$
l_1\cdot l_2=0
$$
for all $l_1\neq l_2\in L_X$.
\end{assumption}
This instantly means that the conditions of \cref{eq:stabilizer_weights} are satisfied for $j\geq 2$, and can be ignored. The down side to this assumption is that it already limits the distance of our code. The distance against $X$ errors, $d_X$, satisfies $d_X\leq |l|$ for all $l\in L_X$. Since all $k$ logicals act on distinct qubits, it must be that
$$
d_Xk\leq n.
$$
If we want $k$ to scale linearly with $n$, $d_X$ cannot be better than a constant. Relaxing this assumption should be the primary target of future work.

\section{Transversal Toolbox}\label{sec:toolbox}

What logical operation should we be aiming for? In this section, we will collect a set of standard circuit-level results that allow us to recreate a wide range of individual logical phase gates from a single logical transversal phase gate. We will also consider the other gates that might be required to yield universality -- two qubit gates, both between distinct code blocks and within a single code block, and a Hadamard gate.

\subsection{Localised Logical Phase}

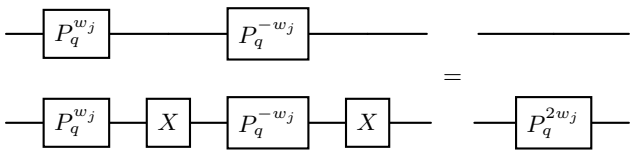
\begin{figure}
\centering
\begin{quantikz}
& \gate{P_q^{w_j}} && \gate{P_q^{-w_j}} &&\midstick[2,brackets=none]{=}&& \\
& \gate{P_q^{w_j}} & \gate{X\vphantom{P_q^{w_j}}} & \gate{P_q^{-w_j}} & \gate{X\vphantom{P_q^{w_j}}} &&\gate{P_q^{2w_j}}&
\end{quantikz}
\caption{In a gate sequence $P_{L,q}^{w},X_L,P_{L,q}^{-w},X_L$, those logical qubits (top row) which do not have the logical $X$ applied cancel, and yield an $\identity$. Those on which $X_L$ is applied (bottom row) create a localised logical gate of twice the phase.}\label{fig:first_circuit}
\end{figure}

\begin{lemma}\label{lem:transversal_localise}
For any code that has transversal $P_{L,q}^{w}$, that same code has gates $P^{w_ie_i}_{L,q-1}=P_{L,q}^{2w_ie_i}$, where $e_i$ is a standard basis vector, creating individual phase gates of double the phase.
\end{lemma}
\begin{proof}
By assumption, we have three ingredients: transversal $P_L=P_{L,q}^{w}$ implemented by $P_q^p$, its inverse $P_L^\dagger\equiv P_q^{2^q-p}$, and individual logical $X$ rotations. We follow the protocol of applying $P_L$, logical $X$ on logical qubit $i$, $P_L^\dagger$, before finally repeating the logical $X$ gate.

For all logical qubits $j\neq i$, the applied circuit is the same as the top row of \cref{fig:first_circuit}, with a net application of $\identity$. For logical qubit $i$, the bottom row of \cref{fig:first_circuit} applies, giving a net phase gate of $P^{2w_i}_q$. Note that $P_q^2=P_{q-1}$.

We can now apply the circuit identities and conclude that if $P_{L,q}^w=P_q^p$, then $P_{L,q}^{w_ie_i}=P_{q-1}^{p\cdot l_i}$ where $l_i$ is the logical $X$ operator corresponding to logical qubit $i$.
\end{proof}

Thus, with $q=4$, we can convert $P_{L,4}^{1,1,\ldots, 1}$, transversal logical $\sqrt{T}$, into individual $T$ gates on any qubit. Combined with Hadamard gates, this gives any arbitrary single qubit rotation. Moreover, the Hadamard is easier to implement via magic state distillation/teleportation than the $T$ gate by virtue of being lower in the Clifford hierarchy. Essentially, it just relies on $\ket{+}$ state preparation, which is just a measurement operation.

So far, we have described how to realise a transversal set of logical phase gates $P_q$, and how these can be localised, within the block, to give transversal $P_{q-1}$ gates. However, we are using CSS codes, which have transversal controlled-\textsc{not} between two code blocks, meaning that this physical operation implements logical controlled-\textsc{not} between matching pairs of qubits of the two copies. We can use these to significantly increase our availability of gates, in order to move toward universality.

\subsection{Localised Logical Phase with Ancillas}

We have just seen how $P_{L,q}^w$ can be used to localise individual $P_{q-1}$ logical phase gates. If we extend our set of allowed operations to include introducing other copies of the code prepared in standard states such as $\ket{0}$ or $\ket{+}$, and the measurement of logical operators, we can use standard tricks from magic state distillation in order to realise an individual logical phase gate $P_q$. What is best in terms of overhead to realise a particular localised logical $P_q$ -- transversal $P_q$ and an ancilla block, or transversal $P_{q-1}$ acting within a single block -- will depend on the details of the individual code.

\begin{lemma}\label{lem:transversal_with_ancilla}
Given transversal $P_{L,q}^w$ and an ancilla block, we can create individual $P_{L,q}^{w_ie_i}$.
\end{lemma}
\begin{proof}
All CSS codes have transversal controlled-\textsc{not} available between different code blocks (see \cref{fig:inter_intra} for a depiction). Use this in the following protocol:
\begin{enumerate}
  \item Prepare an ancilla block in $\ket{+}^{\otimes k}$ by $X$ basis measurement.
  \item Apply $P_{L,q}^w$ to the ancilla block.
  \item Perform an $X$ basis measurement on every logical qubit of the ancilla block except $i$.
  \item Apply transversal c-\textsc{not} controlled by the data block and targetting the ancillas.
  \item Measure the ancilla block in the $Z$ basis.
  \item Apply any necessary phase $P_{L,q-1}$ correction based on measurement results, as specified in the following circuits. (We have these by \cref{lem:transversal_localise}.)
\end{enumerate}
At the end of step 3, every logical qubit of the ancilla block is in state $\ket{+}$ except for qubit $i$, which is in the state $P_q^{w_i}\ket{+}$. For all logical qubits except $i$, the ensuing circuit is thus
$$
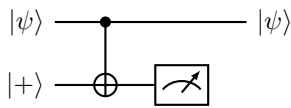
\begin{quantikz}
\lstick{\ket{\psi}} & \ctrl{1} & & \rstick{\ket{\psi}} \\
\lstick{\ket{+}} & \targ{} & \meter{}
\end{quantikz}
$$
meaning that nothing happens to the original input. For logical qubit $i$, the different preparation of the ancilla implies the circuit
$$
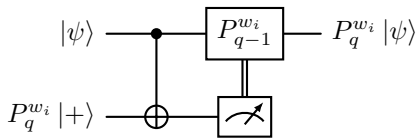
\begin{quantikz}
\lstick{\ket{\psi}} & \ctrl{1} & \gate{P_{q-1}^{w_i}} & \rstick{$P_q^{w_i}\ket{\psi}$} \\
\lstick{$P_q^{w_i}\ket{+}$} & \targ{} &\meter{}\wire[u]{c}
\end{quantikz}
$$
\end{proof}

\subsection{Inter-Block Controlled-Phase}\label{sec:cPhase_inter}

\begin{figure}
\centering
\begin{quantikz}
&&\ctrl{1} &&\ctrl{1} &\midstick[2,brackets=none]{=} &  & \\
 & & \targ{} & & \targ{} & &  & \\
&&\ctrl{1} &&\ctrl{1} &\midstick[2,brackets=none]{=} &  \ctrl{1} & \\
 & \gate{P\vphantom{P^\dagger}} & \targ{} & \gate{P^\dagger} & \targ{} & & \gate{P^2} &
\end{quantikz}
\caption{In a gate sequence of controlled-\textsc{not} (transversal, between two different code blocks), $P_L$, controlled-\textsc{not}, $P_L^\dagger$, then those qubits that are not acted upon by $P_L$ (top row) are unchanged. Those acted upon by $P_L$ (bottom row) gain a net controlled-phase action.}\label{fig:cphase_localise}
\end{figure}
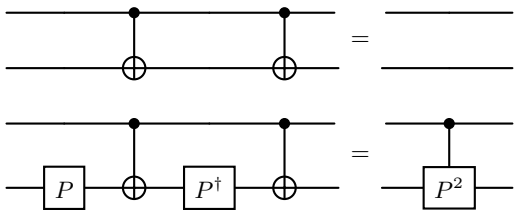

As soon as we can localise an $S$ gate (e.g.\ from logical transversal $T$, $q=3$, without ancillas or $q=2$ with ancillas), we can localise the effect of the transversal controlled-\textsc{not}, allowing us to apply a single logical controlled-phase gate between matching logical qubits of two blocks.

\begin{lemma}\label{lem:cphase}
Given $P_{L,2}^{e_i}$, we can create individual c-\textsc{phase} between the logical qubits $i$ of two code blocks.
\end{lemma}
\begin{proof}
The protocol is as follows:
\begin{enumerate}
\item Apply a transversal controlled-\textsc{not} between two control blocks.
\item On the target block, apply $P_{L,2}^{e_i}$.
\item Apply the same transversal controlled-\textsc{not} between the two control blocks.
\item On the target block, apply $P_{L,2}^{-e_i}$.
\end{enumerate}
On the logical qubits $j\neq i$ (\cref{fig:cphase_localise}, top row, taking $P=S$), the only action is two controlled-\textsc{not}s, which cancel each other. On logical qubits $i$ (\cref{fig:cphase_localise}, bottom row), the net effect is controlled-\textsc{phase} between that pair of logical qubits only. This is just a controlled version of \cref{fig:first_circuit}.

In practical terms, this gate is implemented by applying controlled-phase gates between equivalent pairs of qubits in the two code blocks, identified by $l_i$, the logical $X$ operator of qubit $i$.
\end{proof}

We could see this equally well by directly using the previous tools for the transversal gate analysis, which is exactly what appeared in \cite{webster2023}.

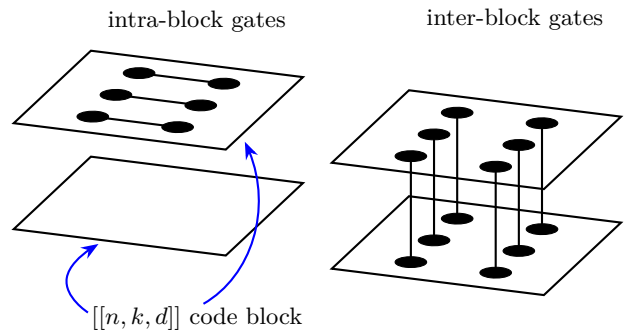
\begin{figure}
\centering
\tdplotsetmaincoords{70}{110}
\begin{tikzpicture}[scale=3,tdplot_main_coords]
\draw[thick,fill=white] (0,0,0) -- (1,0,0) -- (1,1,0) -- (0,1,0) -- cycle; 
\draw[thick,fill=white] (0,0,0.5) -- (1,0,0.5) -- (1,1,0.5) -- (0,1,0.5) -- cycle;

\draw [thick] (0.2,0.3,0.5) -- (0.2,0.7,0.5);
\draw [thick] (0.5,0.3,0.5) -- (0.5,0.7,0.5);
\draw [thick] (0.8,0.3,0.5) -- (0.8,0.7,0.5);
\draw[fill]  (0.2,0.3,0.5) circle [radius=0.07];
\draw[fill]  (0.2,0.7,0.5) circle [radius=0.07];
\draw[fill]  (0.5,0.3,0.5) circle [radius=0.07];
\draw[fill]  (0.5,0.7,0.5) circle [radius=0.07];
\draw[fill]  (0.8,0.3,0.5) circle [radius=0.07];
\draw[fill]  (0.8,0.7,0.5) circle [radius=0.07];

\draw[thick,fill=white] (0,1.5,0) -- (1,1.5,0) -- (1,2.5,0) -- (0,2.5,0) -- cycle; 
\draw[thick,fill=white] (0,1.5,0.5) -- (1,1.5,0.5) -- (1,2.5,0.5) -- (0,2.5,0.5) -- cycle;

\draw [thick] (0.2,1.8,0) -- (0.2,1.8,0.5);
\draw[fill]  (0.2,1.8,0) circle [radius=0.07];
\draw[fill]  (0.2,1.8,0.5) circle [radius=0.07];

\draw [thick] (0.2,2.2,0) -- (0.2,2.2,0.5);
\draw[fill]  (0.2,2.2,0) circle [radius=0.07];
\draw[fill]  (0.2,2.2,0.5) circle [radius=0.07];

\draw [thick] (0.5,1.8,0) -- (0.5,1.8,0.5);
\draw[fill]  (0.5,1.8,0) circle [radius=0.07];
\draw[fill]  (0.5,1.8,0.5) circle [radius=0.07];

\draw [thick] (0.5,2.2,0) -- (0.5,2.2,0.5);
\draw[fill]  (0.5,2.2,0) circle [radius=0.07];
\draw[fill]  (0.5,2.2,0.5) circle [radius=0.07];

\draw [thick] (0.8,1.8,0) -- (0.8,1.8,0.5);
\draw[fill]  (0.8,1.8,0) circle [radius=0.07];
\draw[fill]  (0.8,1.8,0.5) circle [radius=0.07];

\draw [thick] (0.8,2.2,0) -- (0.8,2.2,0.5);
\draw[fill]  (0.8,2.2,0) circle [radius=0.07];
\draw[fill]  (0.8,2.2,0.5) circle [radius=0.07];

\node at (0,0.5,0.7) {intra-block gates};
\node at (0,2,0.9) {inter-block gates};

\node at (0,0.5,-0.7) {$[[n,k,d]]$ code block};

\draw [thick,-Stealth,blue] plot [smooth,tension=1] coordinates {(1,0.35,-0.35) (1,0.25,-0.2) (1,0.4,-0.02)};
\draw [thick,-Stealth,shorten >= 0.1cm,shorten <= 0.1cm,blue] (0,0.5,-0.65) to[out=30,in=-60] (0.8,1,0.5);
\end{tikzpicture}
\caption{For multiple instances (``blocks'') of a qLDPC code, transversal gates can be applied in two ways. Inter-block copies interact different copies For example, all CSS codes have transversal controlled-\textsc{not}. It may also be possible to generate intra-block interactions where only a single copy is required.}\label{fig:inter_intra}
\end{figure}


\begin{corollary}
For a code with transversal $P_q$, this can be used to build individual controlled-$P_{q-2}$ gates without the use of ancillas, or controlled-$P_{q-1}$ gates with ancillas.
\end{corollary}

There is additional potential to create multi-control gates \cite{koutsioumpas2022}, where each additional control will generally cost an additional power of two in the phase, but we will not examine these further here.


\subsection{Intra-Block controlled-Phase with Ancilla}

\begin{figure}
\centering
\begin{quantikz}
\lstick{$\ket{\psi}_{B,i}$} & \ctrl{2} & && \gate{Z}& \\
\lstick{$\ket{\phi}_{B,j}$} &&\ctrl{2}&\gate{Z}&& \\
\lstick[2]{$\frac{1}{\sqrt{2}}(\ket{0+}+\ket{1-})_{(A,i),(A,j)}$} & \targ{} &&\meter{}\wire[u]{c} \\
&&\targ{}&&\meter{}\wire[u][3]{c}
\end{quantikz}
\caption{Circuit for performing controlled-phase between qubits $i$ and $j$ of block $B$, using an ancilla block $A$ prepared and measured in stabilizer states. The controlled-\textsc{not} is the transversal gate available to all CSS codes.}\label{fig:intra_cp}
\end{figure}
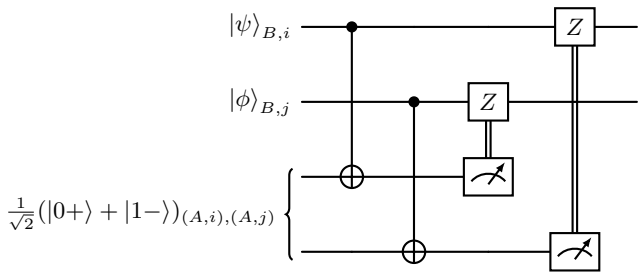

So far, the controlled-phases that we have implemented only act between equivalently labelled qubits within two different blocks, say the logical $i$ of each block. This is not enough for universality -- we need to make differently labelled qubits, say $i$ and $j$, interact as well. We will do this by preparing different stabilizer states within an ancilla block, via measurement of stabilizers. This takes us beyond the standard concept of transversality, and is likely to impact fault-tolerant proofs differently. That said, if we can make codes with sufficiently large distance, the impact should be negligible as correlated pair-wise errors should be irrelevant.

Let $\ket{\psi}$ and $\ket{\phi}$ be different logical qubits, $i$ and $j$, of one block $B$. Consider the matching two qubits $i$ and $j$ of an ancilla block $A$, and prepare them in the two-qubit cluster state \cite{raussendorf2001} (eigenstates of $X\otimes Z$ and $Z\otimes X$). The circuit of \cref{fig:intra_cp} directly implements controlled-phase between $\ket{\psi}$ and $\ket{\phi}$. All other qubits in the ancilla block should be prepared in $\ket{+}$ so that the transversal controlled-\textsc{not} has no effect.

\subsection{Intra-Block controlled-Phase without Ancilla}\label{sec:intra_phase}

While it will be useful to apply controlled-phase gates between two distinct blocks, as promised by \cref{sec:cPhase_inter}, it would be even more efficient to apply such gates between logical qubits within a single block, without an ancilla, as depicted in \cref{fig:inter_intra}.

Consider a system that has logical transversal $T$ (and hence individual $S$ gates). Let us further assume that our system has some symmetry property such that
$$
S_LH_XS_R=H_X.
$$
This means that for every generator $r\in H_X$, then $rS_R\in H_X$. Let us further assume that there's a pair of logical $X$ operators $l_1,l_2$ such that $l_2=S_Rl_1$ and $l_1\cdot l_2=0$. What happens if we apply a set of $|l_1|$ controlled-phase gates between pairs of qubits $(i,j)$ such that $(l_1)_i=1$ and $(l_2)_j=1$? First, let us note that we can essentially ignore all the other logical operators by \cref{assum:unsupported_logicals} as they are unaffected by the controlled-phase gates.

The all-0 logical state is a superposition of states $r$. Each $r$ gains a phase
$
\pi|S_Rr\cdot r\cdot l_2|.
$
If $r\in H_X$, then $y=S_Rr\in H_X$. Since the code possesses transversal $T$, we already have that terms $|H_X\cdot H_X\cdot L_X|=0\text{ mod }2$. Hence, all such terms acquire no phase. By the same token, terms such as $r\oplus l_1$ and $r\oplus l_2$ acquire no phase. On the other hand, a term $r\oplus l_1\oplus l_2$ acquires the phase
$$
|(S_R(r\oplus l_1\oplus l_2))\cdot (r\oplus l_1\oplus l_2)\cdot l_2|=|l_2|
$$
given that $|H_X\cdot H_X\cdot L_X|=0\text{ mod }2$. Hence, we have implemented logical controlled-phase if $|l_2|$ is odd.

\begin{example}Let $H$ be the $X$-parity check of an $[[n,1,d]]$ code with transversal $S$, and assume it has a logical $X$ operator $X_l$. We take two copies of this,
$$
H_X=\begin{bmatrix} H & 0 \\ 0 & H \end{bmatrix}.
$$
This is a $[[2n,2,d]]$ code with logical operators $\begin{bmatrix} l & 0 \end{bmatrix}$ and $\begin{bmatrix} 0 & l \end{bmatrix}$. There is a symmetry operator
$$
S_R=\begin{bmatrix} 0 & \identity \\ \identity & 0 \end{bmatrix}.
$$
($S_L$ is the same, just the $\identity$ matrices are of a different dimension.) We see that $S_R$ maps between the two logical states, and preserves the logical space. These cases all have transversal controlled-$Z$ between the pairs of qubits. This is equivalent to \cref{lem:cphase} for inter-block operations, except that this applies intra-block.
\end{example}

\noindent This directly generalises:

\begin{lemma}
If a code satisfies an appropriate symmetry condition $S_LH_XS_R=H_X$, it has transversal logical $P_q$, and two logical operators $l_1,l_2$ such that $l_1\cdot l_2=0$ and $S_Rl_1=l_2$, then it also has logical controlled-$P_{q-2}$ between logical qubits $l_1$ and $l_2$.
\end{lemma}

We note that the change in $q$ value makes a lot of sense. We lose one for converting transversal to local, and we lose another one converting from one-qubit to two-qubit. Again, this can generalise to multiple controls, with each control further diminishing the $q$ value.


We have now proven that if we can find a single gate $P^w_{L,q}$ where every element of $w$ is odd, we can implement a wide range of phase and controlled-phase gates on a single code block. The only supplement that we would need for a universal set of gates is Hadamard.

\subsection{Hadamard Using Measured Ancilla}

In a code with individual logical controlled-phase gates available, we can use stabilizer measurements to prepare $\ket{+}$ logical states, and perform logical $X$ measurements. We can compose these to create a Hadamard gate, inspired by measurement-based quantum computation \cite{raussendorf2001}:
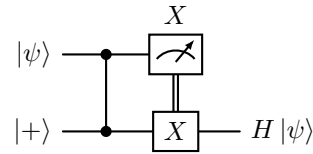
\[
\begin{quantikz}
\lstick{\ket{\psi}} & \ctrl{1} & \meter{X}\wire[d]{c} \\
\lstick{\ket{+}} & \ctrl{0} & \gate{X} & \rstick{$H\ket{\psi}$}
\end{quantikz}
\]
The challenge with this case is that if we're using a separate ancilla block, the Hadamarded qubits are left stranded on a different block. This can be corrected with additional overhead.


\section{qLPDC Code Constructions}

There are increasingly many code constructions available now for qLDPC codes, returning both large numbers of logical qubits and high distance. We choose two relatively straightforward cases, the balanced product and hypergraph product, with which we will demonstrate our constructions in this paper, give explicit descriptions of the logical Pauli operators etc. Both are based on the combination of an expander graph with a local code to form a Tanner code. Essentially any code construction that invokes it (or can choose to invoke it) should have an equivalent manipulation.

The details of expander graphs are unimportant to our present purposes; simply, they provide us with binary matrices with high distance and a large code space. Our main priority is to leave them untouched -- they are sufficiently complex to produce that imposing the extra constraints of \cref{lem:transversal} is too daunting (and in the case of the incidence matrix of a regular graph, cannot be imposed). Instead, we aim to see how properties of the small, local code might imbue the entire code with the desirable properties. Typically, the standard constructions rely upon an expander graph which is an $s$-regular graph, from which we use the incidence matrix. Nothing that we do here requires this, however, and it would be equally valid to utilise the biadjacency matrix of a biregular graph where one of the degrees is $s$.

The key challenge to designing a quantum code is ensuring that all the check operators commute. The first step in doing so is to use a CSS code, such that we separate out the $X$-type parity-check operators from the $Z$-type parity-checks. These are described in binary matrices $H_X$ and $H_Z$ respectively. All $X$-types commute with each other and all $Z$-types commute with each other. Cross-commutation is ensured by finding a case where
$$
H_X\cdot H_Z^T=0\text{ mod }2.
$$
\cref{sec:transversal} conveyed that the explicit form of the parity check operators is essential in identifying transversal gates, so we consider two specific cases in this paper.

\subsection{Hypergraph Product Codes}

The hypergraph product \cite{tillich2014} ensures this commutation. We select two binary matrices (perhaps good expanders) $A$ and $B$, building the parity checks as
\begin{align*}
H_X&=\begin{bmatrix} A\otimes\identity & \identity\otimes B^T \end{bmatrix} \\
H_Z&=\begin{bmatrix} \identity\otimes B & A^T\otimes\identity \end{bmatrix}.
\end{align*}
If $B^T$ and $A^T$ are full rank, then
$$
K=(n_A-m_A)(n_B-m_B)
$$
counts the number of logical qubits, where $A\in\{0,1\}^{m_A\times n_A}$ and $B\in\{0,1\}^{m_B\times n_B}$, while the number of physical qubits involved is $N=n_An_B+m_Am_B$.

We will be considering cases specifically where $A^T$ is a Tanner code combining a local code with an expander graph. This is not a standard part of the construction, but we are free to impose it.

The value of an expander graph can be seen in the distance of the code. The matrix $A$ is said to be an $(\alpha,\beta)$-expander if for all $x$ of weight $|x|\leq\alpha n_A$, $|Ax|\geq\beta|x|$ \footnote{Sometimes the definition of $A$ being expanding can be given as both $A$ and $A^T$ being expanding, in terms of the definition we have just given.}. In particular, if we tried some logical operator of the form
$$
l=\begin{bmatrix} x\otimes e \\ 0 \end{bmatrix},
$$
where $e$ is a standard basis vector ($|e|=1$), then this could only be a logical operator if $x$ is non-trivial and $Ax=0$. The logical operator must have a weight $|x|>\alpha n_A$. While this argument alone is insufficient to prove the distance, as there are other representations of the same operator when dressed by stabilizers, it is the foundation of such analyses \cite{panteleev2022a,leitch2025}.

\subsection{Balanced Product Codes}

The balanced product code \cite{breuckmann2021,leitch2025} is another structure that ensures the commutation property of the stabilizers, but it uses a symmetry property of the matrix $A$ to avoid the need for a tensor product (and the associated cost in qubit numbers). In particular, let $A$ be a matrix with some symmetry such that $RA=AC^T$ where the symmetry operators $R$ and $C$ both have all orbits with length exactly $l$. We thus have the properties that $R^l=\identity$ and $C^l=\identity$, which we use to define the parity check matrices
\begin{align*}
H_X&=\begin{bmatrix} A^T & \identity+C \end{bmatrix} \\
H_Z&=\begin{bmatrix} \identity+R & A \end{bmatrix}.
\end{align*}
Originally \cite{breuckmann2021} the balanced product code incorporated the Tanner code construction. While not necessary \cite{leitch2025}, we will use it here. In that case, we require a further assumption that the edges of the graph determining $A$ do not connect vertices in the same orbit.

Again $A$ has large distance due to its expansion properties, while the other parts, $\identity+C$ and $\identity+R$, describe repetition codes with distance $l$.

\subsection{Tanner Codes}

A Tanner code \cite{tanner1981} is formed by combining a graph code with a local code, and is often used to boost the distance of the code. Imagine that we have a binary matrix where every row has weight $s$. In the low density context, $s$ is constant, while the number of rows/columns scales. Strategies for designing appropriate matrices could include using the incidence matrix of an $s$ regular graph or the biadjacency matrix of a bipartite graph. Call this $I_0$. We also take a local code $C_0$ whose parity-check comprises $r$ rows and $s$ columns.
\begin{example}
Consider the 3-regular graph
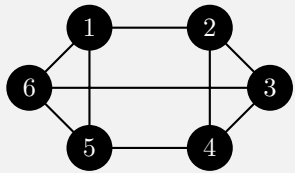
\begin{center}
\begin{tikzpicture}[scale=0.8]
\draw [thick] (0,0) -- (0,-2) (2,0) -- (2,-2) (-1,-1)--(3,-1);
\draw [thick] (0,0) node [circle,fill=black,text=white] {1} -- (2,0) node [circle,fill=black,text=white] {2} -- (3,-1) node [circle,fill=black,text=white] {3} -- (2,-2) node [circle,fill=black,text=white] {4} -- (0,-2) node [circle,fill=black,text=white] {5} -- (-1,-1) node [circle,fill=black,text=white] {6} -- cycle;
\end{tikzpicture}
\end{center}
This has an incidence matrix
$$
I_0=\begin{bmatrix}
1 & 0 & 0 & 0 & 0 & 1 & 1 & 0 & 0 \\
1 & 1 & 0 & 0 & 0 & 0 & 0 & 1 & 0 \\
0 & 1 & 1 & 0 & 0 & 0 & 0 & 0 & 1 \\
0 & 0 & 1 & 1 & 0 & 0 & 0 & 1 & 0 \\ 
0 & 0 & 0 & 1 & 1 & 0 & 1 & 0 & 0 \\ 
0 & 0 & 0 & 0 & 1 & 1 & 0 & 0 & 1
\end{bmatrix}.
$$
Since each row has weight 3, we select a local code of 3 bits, such as the repetition code
$$
C_0=\begin{bmatrix} 1 & 1 & 0 \\ 0 & 1 & 1 \end{bmatrix}.
$$
\end{example}
The new adjacency matrix $A$ is formed by taking $r$ copies of each row of $I_0$. This then includes an $r\times s$ subset which is all-ones (although it may not be in a contiguous block). Replace these $1$s with a copy of $C_0$. The columns of $C_0$ may appear in any order, and there is much freedom in this choice which can influence the code properties. Our analysis will cover all such eventualities.
\begin{example}
To construct the Tanner code, we repeat each row in $I_0$ and replace each $2\times 3$ block with $C_0$ (possibly with permuted columns).
\begin{equation}\label{eq:colourful}A=
\begin{tikzpicture}[baseline=(m-6-1.south)]
\matrix (m) [matrix of math nodes,left delimiter={[},right delimiter={]}] {
1 & 0 & 0 & 0 & 0 & 1 & 0 & 0 & 0 \\
0 & 0 & 0 & 0 & 0 & 1 & 1 & 0 & 0 \\
1 & 0 & 0 & 0 & 0 & 0 & 0 & 1 & 0 \\
1 & 1 & 0 & 0 & 0 & 0 & 0 & 0 & 0 \\
0 & 0 & 1 & 0 & 0 & 0 & 0 & 0 & 1 \\
0 & 1 & 1 & 0 & 0 & 0 & 0 & 0 & 0 \\
0 & 0 & 1 & 1 & 0 & 0 & 0 & 0 & 0 \\ 
0 & 0 & 0 & 1 & 0 & 0 & 0 & 1 & 0 \\
0 & 0 & 0 & 1 & 0 & 0 & 1 & 0 & 0 \\ 
0 & 0 & 0 & 0 & 1 & 0 & 1 & 0 & 0 \\ 
0 & 0 & 0 & 0 & 0 & 1 & 0 & 0 & 1 \\
0 & 0 & 0 & 0 & 1 & 1 & 0 & 0 & 0 \\
};
\begin{scope}[on background layer]
\draw [draw=red,line width=0.5cm,opacity=0.5,line cap=round] (m-1-1.center)--(m-2-1.center) (m-1-6.center) -- (m-2-6.center) (m-1-7.center)--(m-2-7.center);
\draw [draw=green,line width=0.5cm,opacity=0.5,line cap=round] (m-3-1.center)--(m-4-1.center) (m-3-2.center) -- (m-4-2.center) (m-3-8.center)--(m-4-8.center);
\draw [draw=blue,line width=0.5cm,opacity=0.5,line cap=round] (m-5-2.center)--(m-6-2.center) (m-5-3.center) -- (m-6-3.center) (m-5-9.center)--(m-6-9.center);
\draw [draw=yellow,line width=0.5cm,opacity=0.5,line cap=round] (m-7-4.center)--(m-8-4.center) (m-7-3.center) -- (m-8-3.center) (m-7-8.center)--(m-8-8.center);
\draw [draw=purple,line width=0.5cm,opacity=0.5,line cap=round] (m-9-4.center)--(m-10-4.center) (m-9-5.center) -- (m-10-5.center) (m-9-7.center)--(m-10-7.center);
\draw [draw=orange,line width=0.5cm,opacity=0.5,line cap=round] (m-11-5.center)--(m-12-5.center) (m-11-6.center) -- (m-12-6.center) (m-11-9.center)--(m-12-9.center);
\end{scope}
\end{tikzpicture}
\end{equation}
\end{example}
If the initial graph code has a symmetry $R_0I_0=I_0C^T$, then we can choose the columns in such a way that $R=R_0\otimes\identity_r$ and the Tanner code also has a symmetry $RA=C^TA$. This is the specific column ordering that we will use in the case of the balanced product code.


It has been proven \cite{panteleev2022} that if $I_0$ is an expander, then $A$ is also an expander. Moreover, if $C_0$ is full rank, $A^T$ is also an expander. We note, however, that if $C_0$ is not full rank, there are vectors $l$ such that $C_0^Tl=0$. In turn, this means that for every unit vector $v$ (corresponding to a vertex in the original graph), picking out a row in $I_0$, $v\otimes l$ is a null vector of $A^T$. This potentially limits the distance of $A^T$ to being $\leq d(C_0^T)<s$ \footnote{We have not, at this point, proven that these indeed correspond to logical operators of the quantum code, not yet having introduced a quantum code. They could turn out just to be stabilizers and of no concern, but we will see this is not they case.}. This, it will transpire, will be our primary limitation as we will need to introduce such degeneracy for the purposes of facilitating the transversal gates (and is what induces the limitation to \cref{assum:unsupported_logicals}).

\section{Transversal Phase Gates on Transpose Tanner Codes}\label{sec:balanced_transverse}

We will now study how the local code of a Tanner code can enable transversal gates in qLDPC codes. Critically, we need our code to satisfy the conditions of \cref{lem:transversal}, most crucially that the weights of dot products of rows should all be 0 modulo an appropriate power of 2.

Imagine we took $A$ to be a Tanner code in which we started from the incidence matrix $I_0$ of an $s$-regular graph and combined it with a local code $C_0$. If this $A$ were directly part of $H_X$, then when we start looking at $|H_X|$, and indeed certain terms of the $|H_X^{\cdot i}|$ (where we combine rows of $H_X$ comprising elements from the same row of $I_0$ and different rows of $C_0$), we instantly recognise that the rows of $C_0$ must have the same properties, specifically
$$
|C_0^{\cdot i}|=0\text{ mod }2^{q+1-i}.
$$ 
\begin{example}
In \cref{eq:colourful}, if we examine a pair of rows with the same shading, such as
$$
\begin{tikzpicture}
\matrix (m) [matrix of math nodes,left delimiter={[},right delimiter={]}] {
1 & 0 & 0 & 0 & 0 & 1 & 0 & 0 & 0 \\
0 & 0 & 0 & 0 & 0 & 1 & 1 & 0 & 0 \\
};
\begin{scope}[on background layer]
\draw [draw=red,line width=0.5cm,opacity=0.5,line cap=round] (m-1-1.center)--(m-2-1.center) (m-1-6.center) -- (m-2-6.center) (m-1-7.center)--(m-2-7.center);
\end{scope}
\end{tikzpicture}
$$
then everything is 0 outside the region of the local code. Hence if the weight of the dot product of the two rows is to be 0, so must the weight of the dot product of the rows of the local code (which is not true in this case).
\end{example}
In principle, this is easy to achieve, selecting a code such as a 7-bit Hamming (transversal $S$) or Reed-Muller 15-bit (transversal $T$) code. Unfortunately, this does not help when we come to evaluating other parts of $|H_X\cdot H_X|$. Since we are using the incidence matrix of a regular graph, we can consider two vertices (i.e.\ rows of $I_0$) which are joined by an edge. Since they are only joined by 1 edge, the rows of $I_0$ coincide with a 1 element in exactly one column. The corresponding columns of $C_0$ must be non-trivial. Hence, there will always be a choice of row of $H_X$ where the inner product is 1. Either the code cannot be used in this way, or we must set the element of $p$ corresponding to that column to 0. Inevitably, this leads to $p=0$, which is useless. (This is also the explanation for why we cannot just directly use $I_0$, and are attempting to use the Tanner code construction.)

\begin{example}
Returning to \cref{eq:colourful}, we pick an edge of the graph, i.e.\ column of the matrix, such as the first. This edge connects two vertices. For each, there is a set of $r$ rows corresponding to the local code. Pick a row that has a 1 in our chosen column.
$$
\begin{tikzpicture}
\matrix (m) [matrix of math nodes,left delimiter={[},right delimiter={]}] {
1 & 0 & 0 & 0 & 0 & 1 & 0 & 0 & 0 \\
1 & 0 & 0 & 0 & 0 & 0 & 0 & 1 & 0 \\
};
\begin{scope}[on background layer]
\draw [draw=red,line width=0.5cm,opacity=0.5,line cap=round] (m-1-1.center)--(m-1-1.center) (m-1-6.center) -- (m-1-6.center) (m-1-7.center)--(m-1-7.center);
\draw [draw=green,line width=0.5cm,opacity=0.5,line cap=round] (m-2-1.center)--(m-2-1.center) (m-2-2.center) -- (m-2-2.center) (m-2-8.center)--(m-2-8.center);
\end{scope}
\end{tikzpicture}
$$
This is the only column where the two blocks of local code overlap, so the inner product has weight 1. This column cannot be in $p$, the set of qubits on which the transversal gates are applied.
\end{example}

So, we should not use the parity check $A$ of the Tanner code directly in $H_X$. However, the choice of $H_X$ and $H_Z$ is largely arbitrary, and the other will contain $A^T$. What happens if we use an $H_X$ based on $A^T$ (which we call a Transpose Tanner Code)? We choose $p$ to focus entirely on the part of $H_X$ that involves $A^T$, i.e.\ just the left-hand terms, $p=[11\ldots100\ldots 0]$.
\begin{example}
$$
A^T=\begin{tikzpicture}[baseline=(m-5-1.center)]
\matrix (m) [matrix of math nodes,left delimiter={[},right delimiter={]}] {
1 & 0 & 1 & 1 & 0 & 0 & 0 & 0 & 0 & 0 & 0 & 0 \\
 0 & 0 & 0 & 1 & 0 & 1 & 0 & 0 & 0 & 0 & 0 & 0 \\
 0 & 0 & 0 & 0 & 1 & 1 & 1 & 0 & 0 & 0 & 0 & 0 \\
 0 & 0 & 0 & 0 & 0 & 0 & 1 & 1 & 1 & 0 & 0 & 0 \\
 0 & 0 & 0 & 0 & 0 & 0 & 0 & 0 & 0 & 1 & 0 & 1 \\
 1 & 1 & 0 & 0 & 0 & 0 & 0 & 0 & 0 & 0 & 1 & 1 \\
 0 & 1 & 0 & 0 & 0 & 0 & 0 & 0 & 1 & 1 & 0 & 0 \\
 0 & 0 & 1 & 0 & 0 & 0 & 0 & 1 & 0 & 0 & 0 & 0 \\
 0 & 0 & 0 & 0 & 1 & 0 & 0 & 0 & 0 & 0 & 1 & 0 \\
};
\begin{scope}[on background layer]
\draw [draw=red,line width=0.5cm,opacity=0.5,line cap=round] (m-1-1.center)--(m-1-2.center) (m-6-1.center) -- (m-6-2.center) (m-7-1.center)--(m-7-2.center);
\draw [draw=green,line width=0.5cm,opacity=0.5,line cap=round] (m-1-3.center)--(m-1-4.center) (m-2-3.center) -- (m-2-4.center) (m-8-3.center)--(m-8-4.center);
\draw [draw=blue,line width=0.5cm,opacity=0.5,line cap=round] (m-2-5.center)--(m-2-6.center) (m-3-5.center) -- (m-3-6.center) (m-9-5.center)--(m-9-6.center);
\draw [draw=yellow,line width=0.5cm,opacity=0.5,line cap=round] (m-4-7.center)--(m-4-8.center) (m-3-7.center) -- (m-3-8.center) (m-8-7.center)--(m-8-8.center);
\draw [draw=purple,line width=0.5cm,opacity=0.5,line cap=round] (m-4-9.center)--(m-4-10.center) (m-5-9.center) -- (m-5-10.center) (m-7-9.center)--(m-7-10.center);
\draw [draw=orange,line width=0.5cm,opacity=0.5,line cap=round] (m-5-11.center)--(m-5-12.center) (m-6-11.center) -- (m-6-12.center) (m-9-11.center)--(m-9-12.center);
\end{scope}
\end{tikzpicture}
$$
Pick any two rows. Non-trivial intersections occur between at most 1 block of $r$ columns, which are two different columns of $C_0$, such as
$$
\begin{tikzpicture}
\matrix (m) [matrix of math nodes,left delimiter={[},right delimiter={]}] {
1 & 0 & 1 & 1 & 0 & 0 & 0 & 0 & 0 & 0 & 0 & 0 \\
 1 & 1 & 0 & 0 & 0 & 0 & 0 & 0 & 0 & 0 & 1 & 1 \\
};
\begin{scope}[on background layer]
\draw [draw=red,line width=0.5cm,opacity=0.5,line cap=round] (m-1-1.center)--(m-1-2.center) (m-2-1.center) -- (m-2-2.center);
\draw [draw=green,line width=0.5cm,opacity=0.5,line cap=round] (m-1-3.center)--(m-1-4.center);
\draw [draw=orange,line width=0.5cm,opacity=0.5,line cap=round] (m-2-11.center) -- (m-2-12.center);
\end{scope}
\end{tikzpicture}
$$
If $C_0^T$ satisfies \cref{lem:transversal}, then so does $A^T$ (it does not in this example).
\end{example}
Now the rows of $A^T$ comprise blocks of length $r$ which are rows of $C_0^T$. But these blocks are consistent between the different rows. In particular, let us pick two edges of the graph (rows of $A^T$). If the two edges do not have a vertex in common, the inner product of the rows is 0. If the two edges have a vertex in common, their value of $|H_X\cdot H_X|$ is just equal to the inner product $|C_0^T\cdot C_0^T|$ between two rows of $C_0^T$. Thus, if it were the case that
$$
|{C_0^T}^{\cdot i}|=0\text{ mod }2^{q+1-i},
$$
the transpose Tanner code would satisfy all the properties required for a transversal $P_q$ gate. Clearly, we want to set $C_0$ to be the transpose of a 7-bit Hamming or 15-bit Reed-Muller (for example).
\begin{example}
A symmetrised version of the Hamming code is still $[7,4,3]$, and so is its transpose since $C_0=C_0^T$.
$$
C_0=\begin{bmatrix}
1 & 0 & 0 & 1 & 1 & 0 & 1 \\
0 & 1 & 0 & 1 & 0 & 1 & 1 \\
0 & 0 & 1 & 0 & 1 & 1 & 1 \\
1 & 1 & 0 & 0 & 1 & 1 & 0 \\
1 & 0 & 1 & 1 & 0 & 1 & 0 \\
0 & 1 & 1 & 1 & 1 & 0 & 0 \\
1 & 1 & 1 & 0 & 0 & 0 & 1
\end{bmatrix}
$$ 
\end{example}
If we construct $C_0$ just by adding linearly dependent rows to a code that we've previously analysed, then all this does to $H_Z$ is add linearly dependent rows. This does not change the of distance of that code. Hence, the $X$ distance does not change, and we just inherit the robustness properties of the original constructions.

The linear dependence does alter the properties of $H_X$. This is where the troubles arise. We have imposed that $|C_0|=0\text{ mod }2^q$, which is equivalent to
$$
C_0\begin{bmatrix} 1 \\ 1 \\ \vdots \\ 1\end{bmatrix}=0.
$$
Hence, the matrix must have a null vector. This kills the code distance -- $A^T$ is not expanding, and the distance against $Z$ errors is limited it to the distance of $C_0^T$. Nevertheless, we are willing to accept this just to get some non-zero distance against $Z$ errors (and large distance against $X$ errors) as this will be the first class of qLDPC codes with transversal gates.

At this stage, we have ensured that \cref{eq:stabilizer_weights} is satisfied for $j=0$. To determine if the $j\geq 1$ cases are satisfied, we need to know the form of the logical operators. Let $v$ be the length-$s$ all-ones vector. As already argued, $C_0^Tv=0$. (There may be more null vectors, but we won't make use of them here. We are using a subsystem code.) We can use this to construct logical operators of the form
$$
l_Z^{\text{hyper}}=\begin{bmatrix} (e_I\otimes v)\otimes e_B \\ 0 \end{bmatrix},\qquad l_Z^{\text{balanced}}=\begin{bmatrix} e_I\otimes v \\ 0 \end{bmatrix}
$$
where $e_I$ and $e_B$ are unit vectors of appropriate dimensions (corresponding to systems upon which $I_0$ and $B$, or their transposes, act). The $Z$ logicals are not directly relevant to us, but we need them to prove that we indeed have (distinct) logical qubits.

The corresponding logical $X$s are
\begin{equation}\label{eq:logXtrans}
l_X^{\text{hyper}}=\begin{bmatrix} e_I\otimes v\otimes u \\ 0 \end{bmatrix},\qquad l_X^{\text{balanced}}=\begin{bmatrix} \sum_{i=0}^{l-1}R_0^ie_I\otimes v \\ 0 \end{bmatrix}
\end{equation}
For the hypergraph, $u$ is a null vector of $B$, and $e_B$ is a standard basis vector satisfying $u^Te_B=1$. Note that, in general, the different $u$ will act on some of the same sites, so \cref{assum:unsupported_logicals} is not satisfied. (It is still useful for examples, such as where $B$ has a single null vector, but this is an uninteresting regime for application).

For the balanced product, we see that if each $e_I$ is a standard basis vector chosen to be the representative of a different orbit, then each logical operator acts on a distinct set of qubits, satisfying \cref{assum:unsupported_logicals} and, indeed, saturating the corresponding $d_Xk\leq n$ bound. A graph of $m$ vertices yields $m/l$ distinct orbits, and hence $m/l$ different choices of $e_I$ while the number of physical qubits involved is $3ms/2$. Thus, the number of logical qubits is $\order{N/l}$. Each $l_X$ has a weight $sl$, so provided $s$ and $l$ are both odd, the $(i,j)=(0,1)$ case of \cref{lem:transversal} is satisfied. The only cases that we are yet to consider are $j=1,i\geq 1$. Note, however, the $\otimes v$ structure in $l_X$. This means that the logical operator has an all-ones block corresponding to a run of columns that matches up with some $C_0^T$ block. Hence $H_X^{\cdot i}\cdot L_X$ just evaluates $H_X^{\cdot i}$ on a subset of blocks, with $C_0$ ensuring that all the conditions are satisfied for individual blocks.

\begin{example}
Consider the cycle graph of 10 vertices, where each vertex also has an edge to the one directly opposite it. Its $10\times 15$ incidence matrix is
$$
I_0=\begin{bmatrix} \identity_5 & \identity_5 & \identity_5 \\ \identity_5 & Q^{-1} & Q^2 \end{bmatrix}
$$
where $Q$ is the $5\times 5$ cyclic permutation, $Q^5=\identity$. All the rows have weight 3, and there is an overall symmetry $R_0I_0=I_0C^T$ with
$$
R_0=\identity_2\otimes Q,\qquad C=\identity_3\otimes Q^{-1}.
$$
We can then introduce the local code
$$
C_0=\begin{bmatrix} 1 & 1 & 0 \\ 0 & 1 & 1 \\ 1 & 0 & 1 \end{bmatrix}
$$
where in each block-row of $I_0$, each block column is tensored with a different column of $C_0$ to produce $A$. By doing this, we get $R=R_0\otimes\identity_3$ such that $RA=AC^T$. We can now apply the balanced product construction to yield $H_X,H_Z$ on $n=15+3\times 10=45$ qubits.

This code has 5 logical qubits, of which our construction has given 2:
\begin{align*}
l_X&=\begin{bmatrix}e\otimes\begin{bmatrix} 1 & 1 & 1 & 1 & 1 \end{bmatrix}\otimes \begin{bmatrix} 1 & 1 & 1 \end{bmatrix} \\ 0 \end{bmatrix} \\
l_Z&=\begin{bmatrix}e\otimes\begin{bmatrix} 1 & 0 & 0 & 0 & 0 \end{bmatrix}\otimes \begin{bmatrix} 1 & 1 & 1 \end{bmatrix} \\ 0 \end{bmatrix}
\end{align*}
It has $d_X=5$ (the given logicals are not the lowest weight expression) and $d_Z=3$ (limited by the distance of the local code).

This example makes no attempt to present transversal gate properties, for which we need larger graphs. Nevertheless, all the rows of $H_X$, and all those generated by it, have even weight. If we look at a row of $A^T$ (ignoring the right-hand block of $H_X$, since we will set $p=0$ on that part), then
\begin{adjustbox}{max width=\textwidth}
\parbox{\linewidth}{
\begin{align*}
A^T_1&=\begin{tikzpicture}[baseline=(m-1-1).center]
\matrix (m) [matrix of math nodes,left delimiter={[},right delimiter={]},ampersand replacement=\&] {
1 \& 0 \& 1 \& 0 \& 0 \& 0 \& 0 \& 0 \& 0 \& 0 \& 0 \& 0 \& 0 \& 0 \& 0 \& 1 \& 1 \& 0 \& 0 \& \ldots \\
};
\begin{scope}[on background layer]
\draw [draw=red,line width=0.5cm,opacity=0.5,line cap=round] (m-1-1.center)--(m-1-3.center);
\draw [draw=green,line width=0.5cm,opacity=0.5,line cap=round] (m-1-16.center)--(m-1-18.center);
\end{scope}
\end{tikzpicture} \\
l_X&=\begin{tikzpicture}[baseline=(m-1-1).center]
\matrix (m) [matrix of math nodes,left delimiter={[},right delimiter={]},ampersand replacement=\&] {
1 \& 1 \& 1 \& 1 \& 1 \& 1 \& 1 \& 1 \& 1 \& 1 \& 1 \& 1 \& 1 \& 1 \& 1 \& 0 \& 0 \& 0 \& 0 \& \ldots\\
};
\begin{scope}[on background layer]
\draw [draw=red,line width=0.5cm,opacity=0.5,line cap=round] (m-1-1.center)--(m-1-3.center);
\end{scope}
\end{tikzpicture}.
\end{align*}
}
\end{adjustbox}
The $l_x$ isolates length-$r$ sub-blocks of $H_X$, but these blocks are rows of $C_0^T$, so if $C_0^T$ satisfies the transversality conditions, so does $L_X\cdot H_X$.
\end{example}

\subsection{Example with Transversal \texorpdfstring{$S$}{S}}


Balanced product codes work well for examples, generally keeping the number of qubits lower than the corresponding hypergraph product codes.

We want to find an $A$ with appropriate properties, although we will disregard the requirement of expansion properties for the sake of finding a modestly sized example of a (bi)regular graph. To this end, it is helpful to investigate the field of pairwise balanced uniform block designs. These are specified by parameters $(v,k,\lambda)$ and derived parameters $b,r$ ($bk=vr$ and $\lambda(v-1)=r(k-1)$), and will yield $v\times b$ binary matrices where each row has a weight $r$ and each column has a weight $k$, while $\lambda$ specifies that every pair of rows have 1s in a set of exactly $\lambda$ columns. So, if we take $\lambda=1$, this imposes an unnecessary but convenient additional restriction where pairs of rows only coincide in exactly one position. For our initial example, we want $r=7$, which will allow us to use the Hamming $[7,4,3]$ code for $C_0$. For simplicity, we select $k=3$. This gives us a $15\times 35$ matrix, corresponding to the $(15,3,1)$ block design, which is essentially Kirkman's schoolgirl problem. Indeed, there are several distinct solutions \cite{cole1922}, each with different symmetry properties. One can search through the automorphism groups and find instances of an $l=5$ symmetry. In \cref{kirkman}, we give one example, having reordered the vertices in order to emphasise our chosen symmetry.
 \begin{align*}
C_0&=\begin{bmatrix}
1 & 0 & 0 & 1 & 1 & 0 & 1 \\
0 & 1 & 0 & 1 & 0 & 1 & 1 \\
0 & 0 & 1 & 0 & 1 & 1 & 1 \\
1 & 1 & 0 & 0 & 1 & 1 & 0 \\
1 & 0 & 1 & 1 & 0 & 1 & 0 \\
0 & 1 & 1 & 1 & 1 & 0 & 0 \\
1 & 1 & 1 & 0 & 0 & 0 & 1
\end{bmatrix} \\
R_0&=\begin{bmatrix} 0 & 1 & 0 & 0 & 0 \\ 0 & 0 & 1 & 0 & 0 \\ 0 & 0 & 0 & 1 & 0 \\ 0 & 0 & 0 & 0 & 1 \\ 1 & 0 & 0 & 0 & 0\end{bmatrix}\otimes\identity_3 \\
C^T&=\begin{bmatrix} 0 & 1 & 0 & 0 & 0 \\ 0 & 0 & 1 & 0 & 0 \\ 0 & 0 & 0 & 1 & 0 \\ 0 & 0 & 0 & 0 & 1 \\ 1 & 0 & 0 & 0 & 0\end{bmatrix}\otimes\identity_7
.
\end{align*}
We use these to construct the $105\times 35$ matrix $A$, and from there find $H_X$ and $H_Z$. The code is $[[140,16,d_Z=3,d_X=5]]$, with row and column totals being at most 14 (arising from 3 copies of weight 4 columns from $C_0$ and 2 from the $\identity+C$ type term). Of these 16, 3 are in the form we have explicitly described. 
\begin{align*}
X^L_p&=\prod_{i=1}^{35}X_{i+35p} \\
Z^L_p&=Z_{4+35p}Z_{5+35p}Z_{6+35p}
\end{align*}
for $p\in\{0,1,2\}$ a convenient indexing of the logical operators. The logical $X$ operators all act on distinct sets of qubits. Thus, $|L_X^{\cdot i}|=0$ for all $i>1$.

Clearly, we miss out on a lot of the logical qubits, but this is the cost to ensure that we satisfy all the necessary conditions for transversal gates. The important feature is that the number of logicals still scales well with $n$.

We can verify that $|H_X\cdot p|\equiv 0\text{ mod }4$ and that $|H_X\cdot H_X\cdot p|\equiv 0\text{ mod 2}$. Upon application of $S$ on each of the 105 physical qubits of the left-hand block, each of these logical qubits have $|l\cdot p|\equiv 1\text{ mod }4$, as required for logical $S$ to be applied to them. The details of this example can be verified in \cite{kay2025c}. The code distance is obviously low, and we would use three copies of a Steane code for preference, but we can easily scale the $X$ distance. We have made no particular attempt to optimise the choice of initial graph.

\subsection{Example with Transversal \texorpdfstring{$T$}{T}}\label{sec:exampleT}

For a second example, also in the associated Mathematica notebook \cite{kay2025c}, we select the simplest 15-regular graph -- the complete graph of 16 vertices. This has an incidence matrix of $\tilde I$. While it will certainly have some symmetry, such as a permutation of the vertices, it doesn't have a symmetry suitable for the balanced product \footnote{Recall that we require a permutation of both the vertices and the edges with the same order, such that all orbits have the same length $l$. This can be quite restrictive.}. Instead, we build a new matrix with symmetry built in. In this case, take
$$
I_0=\tilde I\otimes\identity_3.
$$
In doing so, we can have a symmetry
\begin{align*}
R_0&=\identity_{16}\otimes\begin{bmatrix} 0 & 1 & 0 \\ 0 & 0 & 1 \\ 1 & 0 & 0 \end{bmatrix} \\
C&=\identity_{120}\otimes\begin{bmatrix} 0 & 0 & 1 \\ 1 & 0 & 0 \\ 0 & 1 & 0 \end{bmatrix}.
\end{align*}
This can then be combined with the 15 bit Hamming code, which we have symmetrised so that its transpose also yields the code.
$$
C_0=\begin{bmatrix}
1 & 0 & 0 & 0 & 1 & 1 & 1 & 0 & 0 & 0 & 1 & 1 & 1 & 0 & 1 \\
 0 & 1 & 0 & 0 & 1 & 0 & 0 & 1 & 1 & 0 & 1 & 1 & 0 & 1 & 1 \\
 0 & 0 & 1 & 0 & 0 & 1 & 0 & 1 & 0 & 1 & 1 & 0 & 1 & 1 & 1 \\
 0 & 0 & 0 & 1 & 0 & 0 & 1 & 0 & 1 & 1 & 0 & 1 & 1 & 1 & 1 \\
 1 & 1 & 0 & 0 & 0 & 1 & 1 & 1 & 1 & 0 & 0 & 0 & 1 & 1 & 0 \\
 1 & 0 & 1 & 0 & 1 & 0 & 1 & 1 & 0 & 1 & 0 & 1 & 0 & 1 & 0 \\
 1 & 0 & 0 & 1 & 1 & 1 & 0 & 0 & 1 & 1 & 1 & 0 & 0 & 1 & 0 \\
 0 & 1 & 1 & 0 & 1 & 1 & 0 & 0 & 1 & 1 & 0 & 1 & 1 & 0 & 0 \\
 0 & 1 & 0 & 1 & 1 & 0 & 1 & 1 & 0 & 1 & 1 & 0 & 1 & 0 & 0 \\
 0 & 0 & 1 & 1 & 0 & 1 & 1 & 1 & 1 & 0 & 1 & 1 & 0 & 0 & 0 \\
 1 & 1 & 1 & 0 & 0 & 0 & 1 & 0 & 1 & 1 & 1 & 0 & 0 & 0 & 1 \\
 1 & 1 & 0 & 1 & 0 & 1 & 0 & 1 & 0 & 1 & 0 & 1 & 0 & 0 & 1 \\
 1 & 0 & 1 & 1 & 1 & 0 & 0 & 1 & 1 & 0 & 0 & 0 & 1 & 0 & 1 \\
 0 & 1 & 1 & 1 & 1 & 1 & 1 & 0 & 0 & 0 & 0 & 0 & 0 & 1 & 1 \\
 1 & 1 & 1 & 1 & 0 & 0 & 0 & 0 & 0 & 0 & 1 & 1 & 1 & 1 & 0
 \end{bmatrix}.
$$
The resulting construction is a $[[1080,232,3]]$ code, and while all the rows generated by $H_X$, by construction, have weight $0\text{ mod }8$, most do not satisfy the $L_X\cdot H_X^{\cdot i}$ conditions. Instead, we focus on a subsystem code of logical operators that do satisfy the conditions. There are 16 of them, all of which are given by \cref{eq:logXtrans} (there are 16 orbits of length $l=3$, so we create 16 logical qubits). These satisfy all the requirements of \cref{lem:transversal} for $q=3$. Stabilizers have a weight of no more that 18, and no qubit is acted upon by more than 16 qubits. Thus, we get logical transversal $T$ and individual $S$ gates, implemented by transversal gates on physical qubits.


\subsection{Intra-Block Controlled-Phase}

In \cref{sec:intra_phase}, we proved that in the presence of suitable symmetries, it is possible to convert transversal $T$ into localised controlled-phase gates between pairs of logical qubits in the same block. Does this apply to transpose Tanner codes?

Consider a graph with incidence matrix $I_0$ (or biadjacency matrix), which has a symmetry such that $P_VI_0=I_0P_E^T$. As previously indicated, the Tanner code can be selected to respect this symmetry, so that $A$ has a symmetry $(P_V\otimes\identity)A=AP_E^T$. We can then construct a larger symmetry of the hypergraph product:
$$
S_L=P_E^T\otimes\identity,\qquad S_R=\begin{bmatrix} P_V^T\otimes\identity\otimes\identity \\ P_E \end{bmatrix}.
$$
The design of $l_X^{\text{hyper}}$ is then such that the $P_V$ permutes between the different $e_I$ basis elements; controlled-phase gates will be possible between the logical qubits specified by membership of the same orbit of $P_V$, if \cref{assum:unsupported_logicals} is satisfied. Thus, our logical qubits partition into sets of qubits which can mutually interact, but cannot interact between the sets. Of course, if $I_0$ derives from a random graph, chances are it won't have any symmetry. However, there are families with symmetries, such as the LPS expanders \cite{lubotzky1988}, which \cite{breuckmann2021} conveys, for an $n$ vertex graph, has orbits of length $\order{n^{1/3}}$.

The balanced product code is a little different. By construction, it already has a symmetry. However, the $l_X^{\text{balanced}}$ are invariant under the action of that symmetry. Nevertheless, if we could find another permutation of the vertices $P_V$, mapping between the orbits of $R_0$, and if $[P_E,C]=0$ then there is a suitable symmetry
$$
S_L=P_E^T,\qquad S_R=\begin{bmatrix} P_V^T\otimes\identity \\ P_E \end{bmatrix}
$$
where $P_E$ maps between the different elements of $l_X^{\text{balanced}}$ by assumption. Again, the orbits of $P_E$ group the logical qubits into sets within which controlled-phase is possible.

\begin{example}
In the instance constructed in \cref{sec:exampleT}, consider a permutation $\tilde P_L$ that cycles through all the vertices of the original graph. There is certainly a corresponding permutation of the edges such that $\tilde P_L\tilde I=\tilde I\tilde P_R^T$. Then $I_0$ has symmetry operators $P_L=\tilde P_L\otimes\identity_3$, $P_R=\tilde P_R\otimes\identity_3$ and $P_R$ commutes with $C$. Thus, there are controlled-phase gates pairwise between any of the 16 logical qubits. This is readily verified.
\end{example}

\section{Distance Balancing}\label{sec:distance_balance}

The transpose Tanner codes have the prospect of a highly asymmetric distance, $d_X\sim l$ and $d_Z\sim \order{1}$. This was already a problem in the original balanced product paper (this was an asymmetry between distances $d_X \sim l$ and $d_Z \sim n$), where it was resolved with distance balancing. Distance balancing, based on \cite{evra2022}, seeks to redress the asymmetry, taking our original parity checks $H_X$ and $H_Z$, and introducing a new classical code (full rank parity check) $H_c$, combining them into a new code using
\begin{equation}\label{eq:balancing_transversal}
\begin{aligned}
\tilde H_Z&=\begin{bmatrix}
H_Z\otimes\identity & 0 \\
\identity\otimes H_c & H_X^T\otimes\identity
\end{bmatrix}\\
\tilde H_X&=\begin{bmatrix}
H_X\otimes\identity & \identity\otimes H_c^T
\end{bmatrix}
\end{aligned}
\end{equation}
This classical code is chosen to be a good LDPC code, meaning that while acting on $n_c$ physical bits, it has $k_c,d_c\sim n_c$ as well. Such codes exist (and are very common for large enough $n_c$). The resulting code encodes $\tilde k=kk_c$ logical qubits, where $k_c=n_c-m_c$, and $m_c$ is the number of rows in $H_c$, assuming that $H_c$ is full rank. The distances of the new code are 
\begin{align*}
\tilde d_X&=d_Xd_c \\
\tilde d_Z&=d_Z
\end{align*}
as a consequence of the new logical operators
$$
\tilde l_X=\begin{bmatrix}l_X\otimes l_c \\0\end{bmatrix}, \qquad \tilde l_Z=\begin{bmatrix} l_Z\otimes e \\ 0 \end{bmatrix},
$$
where $l_Z$ is a logical $Z$ of the original code, $l_c$ is a codeword of the classical code, and $e$ is a unit vector such that $e\cdot l_c=1$. (Of course, one has to prove that there are no lower weight operators.) In the situation where $d_Z> d_X$, one simply selects $d_c=d_Z/d_X$, and the distance is rebalanced. If $d_Z<d_X$, then we should switch around the definitions to achieve a similar effect:
\begin{equation}\label{eq:balancing_not_transversal}
\begin{aligned}
\tilde H_X&=\begin{bmatrix}
H_X\otimes\identity & 0 \\
\identity\otimes H_c & H_Z^T\otimes\identity
\end{bmatrix}\\
\tilde H_Z&=\begin{bmatrix}
H_Z\otimes\identity & \identity\otimes H_c^T
\end{bmatrix}
\end{aligned}
\end{equation}

In either case, we must assess whether the $\tilde H_X$ satisfies the conditions of \cref{lem:transversal}. A term such as $H_X\otimes\identity$ is essentially multiple parallel, independent, copies of $H_X$. So, if $H_X$ already satisfies the dot product properties, $H_X\otimes\identity$ does as well. By selecting $p=[111\ldots 100\ldots 0]$, $\tilde H_X$ of \cref{eq:balancing_transversal} supports transversal gates if the original $H_X$ did. However, the $L_X$ may no longer act on distinct qubits unless we impose that the logicals of $H_c$ act on distinct qubits. Upon imposing that, there is little to be gained by distance balancing.

Alternatively, we could select $H_c$ to be a Tanner code with the same built-in local code, ensuring the correct transversal properties for $H_c$ as well. Then, the individual rows of $H_X\otimes\identity$ and $\identity\otimes H_c$ in \cref{eq:balancing_not_transversal} could certainly satisfy \cref{lem:transversal}. However, the cross terms, where we select a row of $H_X\otimes\identity$ and a row of $\identity\otimes H_c$ are certainly not always 0 mod 2. For this reason, the distance balancing in the useful direction does not preserve the transversal gate property and we cannot use it. It is unclear whether other variants of distance balancing exist which would allow us to do this.


\section{Direct Construction}\label{sec:direct}

Our original supposition was that the local code in the Tanner graph construction might be helpful. However, the best that we've managed is somewhat disappointing in terms of the storage density and distance trade-offs; we aren't even really using the expansion properties. Instead, we can achieve essentially the same thing (with slightly reduced overheads) in a much more straightforward manner, inspired by one of our examples.

For the balanced product code, let $C_0$ be an appropriate $r\times s$ ``local code'' (although we are not performing the same Tanner graph construction) that has the right properties for a given level $q$ transversal phase gate, such as a Hamming code. The transpose does not need those properties. Imagine we want to encode $k$ logical qubits, and have an $X$ distance of $d_X$. We simply set
$$
A^T=\identity_{k d_X}\otimes C_0
$$
(creating many copies of $C_0$) and
$$
C=R_0\otimes\identity_{kr},\qquad R=R_0^T\otimes\identity_{ks}
$$
where $R_0$ is a permutation matrix with order $d_X$. Since $C_0$ is built in block-wise, the rows of the left-hand block of $H_X$ clearly satisfy the required properties for transversal gates. The logical gates can be constructed as \footnote{There are many more logical operators, just not ones we choose to keep for the transversal gates. In particular, the $J_s$ in $L_Z$ may be replaced by a full set of null vectors of $C_0$, which is how we see that $d(C_0)$ upper bounds the $Z$ distance, while $J_s$ in $L_X$ may be replaced by unit vectors, showing that the $X$ distance is upper bounded by $d_X$.}
\begin{align*}
L_X&= J_{d_X}\otimes\identity_k\otimes J_s\\
L_Z&= e\otimes \identity_k\otimes J_s,
\end{align*}
where $J_a$ is an all-ones vector of length $a$ and $e$ is a unit vector. All $L_X$ logical operators act on distinct qubits, so $l_1\cdot l_2=0$, and the use of $J_s$ means that for $l_1\cdot r_1$, this takes runs of $s$ bits in $r_1$, each of which corresponds to a block of $C_0$, and each $C_0$ has even row weight so $|l_1\cdot r_1|=0\text{ mod }2^{q}$. Thus, the code has the required transversal phase gates.

In fact, this is literally just $k$ parallel copies of a single-qubit storage option where we simply have the control to change the $X$ distance quite conveniently, while maintaining the LDPC property. This emphasises the fact that the code doesn't just have the transversal logical $P_q$ gate, it has individual logical $P_q$. Thus, we must also have controlled-($q-1$) phases between individual pairs of qubits. Explicit examples are given in \cite{kay2025c}.

It is also worth observing that the constructed code is also a hypergraph product code. As such, we won't give a separate version of the construction, which essentially boils down to the same thing. This construction is simple enough that we can build an example with $q=3$ and $k=2$, and explicitly verify, by enumerating every possible basis state, that it does indeed have the claimed transversal $T$ gate, in apparent contradiction with \cite{burton2022}.


\subsection{Hypergraph No-Go}

In \cite{burton2022}, it was stated that when ``robust'' codes $A$ and $B$ are used to form a hypergraph product code, then it is impossible to have non-Clifford transversal logical operators. Since we claim to have bypassed this, how have we achieved that? The key lies in the definition of robust, which has two different forms in \cite{burton2022}. The primary definition, from which the main theorem derives, requires that an $m\times n$ full rank parity check matrix $A$ should have two distinct sets of $k=n-m$ qubits that satisfy certain properties. Note, however, that given there are only $n$ qubits, if $k>n/2$, then there cannot be two distinct sets, and the code is not robust. In this case, it may be possible to have transversal non-Clifford gates. It is subsequently claimed in \cite{burton2022} that an equivalent definition of robustness is that if (up to column permutations)
$$
A=\begin{bmatrix} J^T & \identity \end{bmatrix}
$$
then $J$ should be full rank. Somewhat misleadingly, our examples appear to fulfil this definition of robustness! The resolution comes in an implicit assumption in the definition that $k<n/2$. Hence this secondary characterisation does not apply in our case.

Let us show that the cases we have examined in this paper are not robust, and therefore the no-go result of \cite{burton2022} does not apply. For the transpose Tanner codes of \cref{sec:balanced_transverse}, assume that we have some initial code $A_0$ which is an $m\times n$ matrix with row totals $s$ and column totals $c$, meaning that $ms=nc$. Now take the Tanner code of it with a matrix $C_0$ which is $s\times s$. This yields an $ms\times n$ parity-check matrix $A$, and the transpose, which is the one that we use for our code, is $n\times ms=n\times nc$. If $c> 2$, this means we must be working in the regime where the code is not robust ($k\geq n(c-1)> nc/2$) and the no-go theorem doesn't apply.

For the direct construction of \cref{sec:direct}, the hypergraph product version simply takes the form $A=\identity\otimes C_0$ for some $C_0\in\{0,1\}^{r\times s}$ satisfying the transversal gate properties, and $B$ can be any expanding graph. All the $C_0$ matrices, such as the Hamming codes, that we have used in our examples satisfy $r<s/2$, meaning the no-go theorem doesn't apply. More generally, there is no doubting that if $C_0$ satisfies the transversal properties, then so do multiple parallel copies \footnote{It was proven in \cite{koutsioumpas2022} that the smallest $q$-orthogonal code with distance at least 3 has $s=2^{q+1}-1$ columns (such cases have $r=q+1$ rows), and our conditions for transversality, \cref{lem:transversal}, include the $q$-orthogonal conditions as a subset. There could be slightly larger codes with many more rows, but the large gap between the number of rows and columns for the smallest possible case and $q\geq 3$ would suggest that this is unlikely.}!

\section{Summary}

In this paper, we have given the first constructions of qLDPC codes that have transversal gates which are not Clifford. Any $P_q$ gate can be realised, with a corresponding cost in the `low density' nature. The code distance necessarily displays a significant asymmetry between $X$ and $Z$ errors that we were unable to correct with distance balancing. Nevertheless, this shows that there are cases which sidestep the no-go result of \cite{burton2022}, and hopefully therefore represents a critical first step towards qLDPC codes with transversal gates and better error correction properties. Independently from this, we have collected together a toolbox of, assuredly mostly not novel, steps that help us understand what the target gates are for transversal gates and how they can be synthesised, with minimal additional resources, into a universal quantum computation. We have demonstrated, in particular, how the availability of a single transversal phase gate $P_{L,q}^{\{1,1,1,\ldots,1\}}=P_q^p$ implies the availability of a whole wealth of individual logical phase gates, and multi-controlled phase gates (depending on $q$).

What is the performance of the respective codes? For the balanced product code, if we start from an $s$-regular graph of $m$ vertices, the final code will involve $3sm/2$ physical qubits. If its symmetry has a period of $l$, then it has $\order{sm/l}$ logical qubits. The distances are $d_X\sim \order{l}$ and $d_Z\sim \order{1}$. For instance, \cite{breuckmann2021} describes the family of LPS-expanders (Lubotzky, Phillips, Sarnak) which have $l\sim m^{1/3}$. Unfortunately, they are limited to values of $s$ where $s-1$ is prime, while we require odd $s$. It should be possible to take, for example, $s=30$, and use two copies of the 15-bit code. While we have not explicitly considered this option, we can update \cref{eq:logXtrans}, replacing the $v$ with two different instances, the all-ones vectors across each of the two copies.

However, it should be noted that the good classical code is not really doing anything for us. Instead, we might as well use the more direct, compact, constructions of \cref{sec:direct}, which enable us to arbitrarily tune the trade-off between $k$ and $d_X$. The hypergraph product code construction exhibits an equivalent efficiency.

Now that we know that instances of qLDPC codes with interesting transversal phase gates exist, and what to look for, the challenge for the future is to find instances with much better distance performance.


This work was supported by the Engineering and Physical Sciences Research Council [grant number EP/Y004507/1]. We would like to thank the research group of Dan Browne at UCL for useful conversations.

%


\appendix\onecolumngrid
\crefalias{section}{appendix}
\section{Incidence Matrix for Kirkman's Schoolgirl Problem}\label{kirkman}

\begin{align*}
I_0&=
\left[
\begin{array}{ccccccccccccccccccccccccccccccccccc}
 1 & 0 & 0 & 0 & 1 & 1 & 0 & 0 & 1 & 0 & 1 & 0 & 0 & 0 & 0 & 1 & 0 & 0 & 0 & 0 & 1 & 0 & 0 & 0 & 0 & 0 & 0 & 0 & 0 & 0 & 0 & 0 & 0 & 0 & 0 \\
 1 & 1 & 0 & 0 & 0 & 0 & 1 & 0 & 0 & 1 & 0 & 1 & 0 & 0 & 0 & 0 & 1 & 0 & 0 & 0 & 0 & 1 & 0 & 0 & 0 & 0 & 0 & 0 & 0 & 0 & 0 & 0 & 0 & 0 & 0 \\
 0 & 1 & 1 & 0 & 0 & 1 & 0 & 1 & 0 & 0 & 0 & 0 & 1 & 0 & 0 & 0 & 0 & 1 & 0 & 0 & 0 & 0 & 1 & 0 & 0 & 0 & 0 & 0 & 0 & 0 & 0 & 0 & 0 & 0 & 0 \\
 0 & 0 & 1 & 1 & 0 & 0 & 1 & 0 & 1 & 0 & 0 & 0 & 0 & 1 & 0 & 0 & 0 & 0 & 1 & 0 & 0 & 0 & 0 & 1 & 0 & 0 & 0 & 0 & 0 & 0 & 0 & 0 & 0 & 0 & 0 \\
 0 & 0 & 0 & 1 & 1 & 0 & 0 & 1 & 0 & 1 & 0 & 0 & 0 & 0 & 1 & 0 & 0 & 0 & 0 & 1 & 0 & 0 & 0 & 0 & 1 & 0 & 0 & 0 & 0 & 0 & 0 & 0 & 0 & 0 & 0 \\
 1 & 0 & 0 & 0 & 0 & 0 & 0 & 0 & 0 & 0 & 0 & 0 & 1 & 0 & 1 & 0 & 0 & 0 & 0 & 0 & 0 & 0 & 0 & 1 & 0 & 1 & 0 & 0 & 0 & 1 & 1 & 0 & 0 & 0 & 0 \\
 0 & 1 & 0 & 0 & 0 & 0 & 0 & 0 & 0 & 0 & 1 & 0 & 0 & 1 & 0 & 0 & 0 & 0 & 0 & 0 & 0 & 0 & 0 & 0 & 1 & 1 & 1 & 0 & 0 & 0 & 0 & 1 & 0 & 0 & 0 \\
 0 & 0 & 1 & 0 & 0 & 0 & 0 & 0 & 0 & 0 & 0 & 1 & 0 & 0 & 1 & 0 & 0 & 0 & 0 & 0 & 1 & 0 & 0 & 0 & 0 & 0 & 1 & 1 & 0 & 0 & 0 & 0 & 1 & 0 & 0 \\
 0 & 0 & 0 & 1 & 0 & 0 & 0 & 0 & 0 & 0 & 1 & 0 & 1 & 0 & 0 & 0 & 0 & 0 & 0 & 0 & 0 & 1 & 0 & 0 & 0 & 0 & 0 & 1 & 1 & 0 & 0 & 0 & 0 & 1 & 0 \\
 0 & 0 & 0 & 0 & 1 & 0 & 0 & 0 & 0 & 0 & 0 & 1 & 0 & 1 & 0 & 0 & 0 & 0 & 0 & 0 & 0 & 0 & 1 & 0 & 0 & 0 & 0 & 0 & 1 & 1 & 0 & 0 & 0 & 0 & 1 \\
 0 & 0 & 0 & 0 & 0 & 1 & 0 & 0 & 0 & 0 & 0 & 0 & 0 & 0 & 0 & 0 & 0 & 0 & 1 & 1 & 0 & 1 & 0 & 0 & 0 & 1 & 0 & 0 & 0 & 0 & 0 & 0 & 1 & 0 & 1 \\
 0 & 0 & 0 & 0 & 0 & 0 & 1 & 0 & 0 & 0 & 0 & 0 & 0 & 0 & 0 & 1 & 0 & 0 & 0 & 1 & 0 & 0 & 1 & 0 & 0 & 0 & 1 & 0 & 0 & 0 & 1 & 0 & 0 & 1 & 0 \\
 0 & 0 & 0 & 0 & 0 & 0 & 0 & 1 & 0 & 0 & 0 & 0 & 0 & 0 & 0 & 1 & 1 & 0 & 0 & 0 & 0 & 0 & 0 & 1 & 0 & 0 & 0 & 1 & 0 & 0 & 0 & 1 & 0 & 0 & 1 \\
 0 & 0 & 0 & 0 & 0 & 0 & 0 & 0 & 1 & 0 & 0 & 0 & 0 & 0 & 0 & 0 & 1 & 1 & 0 & 0 & 0 & 0 & 0 & 0 & 1 & 0 & 0 & 0 & 1 & 0 & 1 & 0 & 1 & 0 & 0 \\
 0 & 0 & 0 & 0 & 0 & 0 & 0 & 0 & 0 & 1 & 0 & 0 & 0 & 0 & 0 & 0 & 0 & 1 & 1 & 0 & 1 & 0 & 0 & 0 & 0 & 0 & 0 & 0 & 0 & 1 & 0 & 1 & 0 & 1 & 0
\end{array}
\right]
\end{align*}

\end{document}